\documentclass[letterpaper, 10 pt, conference]{ieeeconf}  % Comment this line out if you need a4paper
\IEEEoverridecommandlockouts                              % This command is only needed if 
\usepackage{graphicx}      % include this line if your document contains figures
\usepackage[tbtags]{amsmath}
\usepackage{mathrsfs}
\usepackage{amsfonts}
\usepackage{amssymb}
\usepackage{color}
\let\labelindent\relax
\usepackage{enumitem}
\usepackage{hyperref}
\usepackage{lipsum}
\usepackage{bbm}
\usepackage{todonotes}
\usepackage{mathbbol}
\usepackage{dsfont}
\usepackage{subcaption}
\def\BibTeX{{\rm B\kern-.05em{\sc i\kern-.025em b}\kern-.08em
		T\kern-.1667em\lower.7ex\hbox{E}\kern-.125emX}}

\newcommand{\norm}[1]{\left\lVert#1\right\rVert}

\usepackage{titlesec}
\usepackage{lipsum}% just to generate text for the example

\titlespacing*{\section}
{0pt}{6pt}{4pt}
\titlespacing*{\subsection}
{0pt}{6pt}{4pt}
\def\rea{\mathds{R}}

\newtheorem{remark}{Remark}
\newtheorem{theorem}{Theorem}
\newtheorem{assumption}{Assumption}
\newtheorem{corollary}{Corollary}
\newtheorem{proposition}{Proposition}

\title{\Large \bf Exponential Stability and Tuning for a Class of Mechanical Systems}

\author{Carmen Chan-Zheng, Pablo Borja, Nima Monshizadeh, and Jacquelien M.A. Scherpen
	\thanks{The work of Carmen Chan-Zheng is supported by the University of Costa Rica. The authors are with the Jan C. Willems Center for Systems and 	Control, Engineering and Technology Institute Groningen, Faculty of Science and Engineering, University of Groningen,  9747 AG Groningen, The Netherlands (email: c.chan.zheng@rug.nl,l.p.borja.rosales@rug.nl, n.monshizadeh@rug.nl, j.m.a.scherpen@rug.nl). }
}
\begin{document}
	
	\maketitle
	\thispagestyle{empty}
	\pagestyle{empty}

	%%%%%%%%%%%%%%%%%%%%%%%%%%%%%%%%%%%%%%%%%%%%%%%%%%%%%%%%%%%%%%%%%%%%%%%%%%%%%%%%
	\begin{abstract}
	In this paper, we prove the exponential stability property of a class of mechanical systems represented in the port-Hamiltonian framework.  To this end, we propose a Lyapunov candidate function different from the Hamiltonian of the system.  Moreover,  we study how the proposed analysis can be used to determine the exponential stability and the rate of convergence of some (nonlinear)-mechanical systems stabilized by a \textit{passivity-based control} technique, namely, PID passivity-based control. We implement such a control approach to stabilize a three-degree-of-freedom robotic arm at the desired equilibrium point to illustrate the mentioned analysis.
	\end{abstract}

	%%%%%%%%%%%%%%%%%%%%%%%%%%%%%%%%%%%%%%%%%%%%%%%%%%%%%%%%%%%%%%%%%%%%%%%%%%%%%%%%
	\section{Introduction}
	The port-Hamiltonian (pH) framework is suitable to represent a wide variety of nonlinear physical systems from different domains \cite{duindam2009modeling,van2014port}. This framework highlights the physical properties of the system under study, particularly the roles that the dissipation and the energy play in its behavior. Moreover, the passivity property of these systems is verified by selecting the Hamiltonian as the storage function. While the passivity property can be related to certain types of stability, for instance, $\mathcal{L}_{2}$\textit{-gain stability} and \textit{Lyapunov stability} \cite{van2000l2}, its relationship with stronger stability properties such as \textit{exponential stability} (ES) is not straightforward. In this work, we focus on the study of the ES property of mechanical systems.  
	
	Proving ES properties for nonlinear systems is, in general, a challenging task. In particular, for mechanical systems, the non-constant inertia matrix represents an obstacle to prove such a stability property. To overcome this, the partial linearization via change of coordinates (PLvCC) \cite{venkatraman2010speed} represents a suitable tool as the inertia matrix becomes constant for the transformed system. By adopting this approach, the authors in \cite{romero2014globally, ferguson2019kinetic} prove \textit{global exponential stability} (GES) properties for fully-actuated mechanical systems in closed-loop with passivity-based control (PBC) approaches. Also, the complexity of proving GES properties increases when the controlled mechanical system is underactuated since the controller is unable to modify the damping of the unactuated coordinates. For example, the results reported in \cite{romero2014globally} are only valid when the damping matrix of the closed-loop system is positive definite, which is not the case if the natural damping is neglected. Some additional results for ES properties of underactuated mechanical systems can be found in \cite{gomez2004physical,venkatraman2008control}.

	Customarily, the Hamiltonian is chosen as the Lyapunov candidate function to prove the stability properties of a pH system. While this approach is convenient to prove the effectiveness of PBC techniques to stabilize mechanical systems, it is, in principle, not adequate to prove ES properties for the resulting closed-loop system. This drawback arises because the Hamiltonian is not \textit{negative definite} as there is no damping related to the dynamics of the generalized coordinates. An alternative to address this issue is implementing the PLvCC methodology that preserves the pH structure of the system as reported in \cite{romero2014globally}. 
	
	Another approach to prove ES properties of mechanical systems consists of finding a new Lyapunov function different from the Hamiltonian as it is proposed \cite{acosta2009,delgado2014overcoming}.  However, the proposed Lyapunov candidate in \cite{acosta2009} only proves stability for fully-actuated systems, while the proposed function in \cite{delgado2014overcoming} only proves asymptotic stability.  Moreover, the mentioned references do not explore the effect of the control gains--which may be related to physical quantities--on the rate of convergence of the trajectories of the closed-loop system. Albeit in \cite{ferguson2019kinetic} investigates the effect of the damping in the rate of convergence of the controlled mechanical system, the relationship between the physical quantities of the closed-loop system--or open-loop if the system is uncontrolled--with the decay rate of its trajectories remains rather unexplored.

	In this paper, we employ both strategies, i.e., PLvCC plus finding a suitable Lyapunov candidate function. This approach is suitable to prove the ES properties of a large class of mechanical systems--including underactuated mechanical systems---while preserving the mechanical structure for the closed-loop system. By preserving this structure, we can understand the effect of the damping and the energy of the system on the rate of convergence of its trajectories.

	The \textit{main contributions} of this work are
	\begin{enumerate}[label=(\roman*), wide=0pt]
		\item A novel analysis approach to prove ES properties for a class of pH systems.
		\item The proof of, under mild conditions, the desired equilibrium point for mechanical systems stabilized via PID-PBC is exponentially stable.
		\item An analysis of the effect of modifying the damping and the energy of the closed-loop system on the rate of convergence of its trajectories. 
	\end{enumerate}  
	
	The remainder of this paper is structured as follows: in Section II, we present the analysis to determine the ES properties of a class of pH systems. In Section III, we show how the mentioned analysis is suitable for proving ES properties for nonlinear mechanical systems stabilized via PID-PBC. In Section IV, we present tuning guidelines that relate the PID-PBC parameters to the decay rate of the closed-loop system. In Section V, we illustrate the applicability of the tuning guidelines by showing the experimental results of the stabilization of a robotic arm via PID-PBC. We finalize this manuscript with some concluding remarks in Section VI. 
	
	\textbf{Notation}: We denote the $n\times n$ identity matrix as $I_n$ and the $n\times m$ matrix of zeros as $0_{n\times m}$. For a given smooth function $f:\rea^n\to \rea$, we define the differential operator $\nabla_x f:=(\frac{\partial f}{\partial x})^\top$ and $\nabla^2_x f:=\frac{\partial^2 f}{\partial x^2}$. For a smooth mapping $F:\rea^n\to\rea^m$, we define the $ij-$element of its $n\times m$ Jacobian matrix as $(\nabla_x F)_{ij}:=\frac{\partial F_i}{\partial x_j}$. For a given  matrix ${A\in\rea^{n\times n}}$, we denote its symmetric part by ${A_{sym}:=\frac{1}{2}(A+A^\top)}$. For a given vector $x\in\rea^{n}$, we say that $A$ is \textit{positive definite (semi-definite)}, denoted as $A>0$  ($A\geq0$), if $A=A^{\top}$ and $x^{\top}Ax>0$ ($x^{\top}Ax\geq0$) for all $x\in \rea^{n}-\{0_{n} \}$ ($\rea^{n}$). For a positive (semi-)definite  matrix $A$,  we define the weighted Euclidean norm as $\lVert x \rVert_{A}:=\sqrt{x^{\top}Ax}$. For $A=A^\top$, we denote by $\lambda_{\max}(A)$ as the maximum eigenvalue of $A$.  All the functions considered in this manuscript are assumed to be (at least) twice continuously differentiable.\\ \textbf{Caveat:} when it is clear from the context, we omit the arguments and the subindex in $\nabla$ to simplify the notation.

	\section{Exponential stability of a class of pH systems}\label{main}
	Consider a pH system whose dynamics are described by
	\begin{equation}\label{phsys}
		\arraycolsep=0.4pt \def\arraystretch{1.4}
		\begin{array}{ll}
			\begin{bmatrix}
				\dot{{q}}\\\dot{{p}}
			\end{bmatrix}=\begin{bmatrix}
				0_{n\times n}&{A}({q})\\
				-{A}^{\top}({q}) &\ \ {J}({q},{p})-{D}({q},{p})
			\end{bmatrix}\begin{bmatrix}
				\nabla_{q}	{H}({q},{p}) \\\nabla_{{p}} 	{H}({q},{p})
			\end{bmatrix}\\
			{H}({q},{p})=\frac{1}{2}{p}^\top {p}+{U}({q})
		\end{array}
	\end{equation}
	where $q,p\in \rea^{n}$, $A:\rea^n\to\rea^{n\times n}$ is full rank, $J:\rea^{n}\times \rea^{n}\to\rea^{n\times n}$ is skew-symmetric, $D:\rea^{n}\times \rea^{n}\to\rea^{n\times n}$ is positive definite, $H:\rea^{n}\times \rea^{n}\to\rea_+$ is the Hamiltonian of the system, and $U:\rea^{n}\to\rea_+$ is potential energy of the system.
	
	The following assumptions characterize the class of systems under study throughout this paper:
	\begin{itemize}[wide=0pt]
		\item	\begin{assumption}\label{ass1}
			$U(q)$ is \textit{locally strongly convex} and has an \textit{isolated local minimum} at $0_n$. 
		\end{assumption}
		\item	\begin{assumption}\label{ass2}
			For all $q\in \rea^n$, every element of $A(q)$ is bounded, i.e., $\norm{A(q)}< \infty$. Furthermore, for all $q,p \in \rea^n$, every element of $D(q,p)$ is bounded, i.e., ${\norm{D(q,p)}< \infty}$.
		\end{assumption}
	\end{itemize}
	
	The following theorem establishes the main results of this paper.
	\begin{theorem}\label{prop1}
		Consider the pH in \eqref{phsys} and $x_\star:=0_{2n}$, then:
		\begin{enumerate}[label=\textit{(\roman*)},wide=0pt]
			\item $x_\star$ is an exponentially stable equilibrium point for \eqref{phsys} with Lyapunov function
			\begin{equation} \label{candidate}
				\begin{split}
					S(q,p)&:=H(q,p)+\epsilon{p}^\top \Phi(q)\nabla_q U(q),
				\end{split}
			\end{equation}
			where $\epsilon>0$, and $\Phi:\rea^n \to \rea^{n\times n}$ satisfies
			\begin{equation}\label{condprop1}
				A(q)\Phi(q)+\Phi^\top(q) A^\top(q) >0,
			\end{equation}
			and $\norm{\Phi(q)}\leq \infty$.
			\item Furthermore, $x_\star$ is a \textit{globally} exponentially stable equilibrium point for \eqref{phsys} if $U(q)$ is radially unbounded. 
		\end{enumerate}
	\end{theorem}
	\begin{proof}
		To prove $(i)$, let $x:=col(q,p)$. As a consequence of Assumption \ref{ass1}, there exists a neighborhood of $x_{\star}$ such that (see \cite{nesterov2013introductory}, \cite{boyd2004convex})
		\begin{equation}\label{H}
			\frac{\beta_{min}}{2} \norm{x}^2\leq H \leq\frac{\beta_{\max}}{2}\norm{x}^2,
		\end{equation}
		where $\beta_{min},\beta_{\max}>0$. Define $C(q,p):=\epsilon{p}^\top \Phi(q)\nabla_q U(q)$. Then, it follows that~\footnote{We have used Young's inequality to obtain the third expresion.}
		\begin{equation}\label{C}
			\begin{split}
				C\leq \norm{\epsilon{p}^\top \Phi\nabla_q U}
				%		&\leq \epsilon \norm{{p}}\norm{\Phi}\norm{\nabla_q U}\\
				%		&\leq \frac{\epsilon \norm{\Phi}}{2} (\norm{{p}}^2 + \norm{\nabla_q U}^2)\\
				%		&=\frac{\epsilon \norm{\Phi}}{2} \norm{\nabla H}^2\\
				\leq\frac{\epsilon\norm{\Phi}\beta_{\max}^2}{2} \norm{x}^2.
			\end{split}
		\end{equation}
		Hence, from \eqref{H} and \eqref{C}, we get that 
		\begin{equation}\label{cond1}
			\begin{split}
				&k_1\norm{x}^2\leq S\leq k_2\norm{x}^2
			\end{split}
		\end{equation}
		where $k_1:=\frac{\beta_{min}-\epsilon\norm{\Phi}\beta_{\max}^2}{2}$, $k_2:=\frac{\beta_{max}+\epsilon\norm{\Phi}\beta_{\max}^2}{2}$ are positive constants. Furthermore, $\dot{S}=\dot{H}+\dot{C}$, where ${\dot{H}= -{p}^\top D {p}}$ and 
		\begin{equation*}
			\begin{split}
				\dot{C}= \epsilon&[-\nabla ^\top_q U A \Phi\nabla _q U+{p}^\top(J-D)^\top \Phi \nabla_q U\\
				&+{p}^\top \Phi\nabla^2_q U A{p}+{p}^\top \dot{\Phi}\nabla_q U].
			\end{split}
		\end{equation*}
		Thus,
		\begin{equation}\label{vdot}
			\begin{split}
				\dot{S}&= -\nabla^\top H \Upsilon\nabla H,\\
			\end{split}
		\end{equation}
		where the matrix $\Upsilon(q,p)$ is defined as
		\begin{equation}
			\Upsilon:=  \begin{bmatrix}
				\epsilon A\Phi &  0_{n\times n}\\ 
				\epsilon(J+D) \Phi-\epsilon\dot{\Phi}&D-\epsilon \Phi\nabla^2_q UA
			\end{bmatrix}.
		\end{equation}
		Recall that $H(q,p)$ has a locally isolated minimum at $x_\star$. Then, \eqref{vdot} implies that $\dot{S}(q,p)$ is \textit{locally negative definite} if there exist $\Phi(q),\epsilon$ such that $\Upsilon_{sym}(q,p)$ is positive definite. To prove the existence of such a pair, select $\Phi(q)=A^{\top}(q)$. Hence,
		\begin{equation}\label{upsilonsym2}
			\arraycolsep=0.8pt \def\arraystretch{1.4}
			\begin{array}{lll}
				\Upsilon_{sym}&= \begin{bmatrix}
					\Upsilon_{11}&\Upsilon_{12}\\
					\Upsilon_{12}^\top&\Upsilon_{22}
				\end{bmatrix},\\
				\Upsilon_{11}(q)&:=\epsilon (A(q)A(q)^\top+A(q)^\top A(q)),\\
				\Upsilon_{12}(q,p)&:=\frac{\epsilon}{2}[A(q)[D(q,p)-J(q,p)]-\dot{A}(q)],\\
				\Upsilon_{22}(q,p)&:=D(q,p)-\epsilon(A(q)^\top \nabla_q^2U(q)A(q)).
			\end{array}
		\end{equation} 
		To verify the sign of \eqref{upsilonsym2}, we employ a Schur complement analysis, i.e.,  $\Upsilon_{sym}(q,p)$ is positive definite if and only if its block (1,1) is positive definite and its Schur complement is positive definite, see \cite{horn2012matrix}. Note that $\Upsilon_{11}(q,p)>0$ and there exists a \textit{sufficiently small} $\epsilon>0$ such that the Schur complement of $\Upsilon_{sym}(q,p)$ is also positive definite. Therefore, for such an $\epsilon$, $\Upsilon_{sym}(q,p)$ is positive definite and $\dot{S}(q,p)$ is locally negative definite. Moreover, let $\mu(q,p)>0$ be the minimum eigenvalue of $\Upsilon_{sym}(q,p)$, then, it follows that
		\begin{equation}\label{vcond2}
			\dot{S}\leq -\mu \norm{\nabla H}^2\leq -\mu\beta_{\max}^2 \norm{x}^2.
		\end{equation}
		Therefore, from \eqref{cond1} and \eqref{vcond2}, $x_\star$ is a \textit{locally} exponentially stable equilibrium point for \eqref{phsys} (see Theorem 4.10 of \cite{khalil2002nonlinear}).
		
		To prove $(ii)$, note that if $U(q)$ is radially unbounded, then the Lyapunov candidate $S(q,p)$ is also radially unbounded, i.e, $S\to \infty$ as $\norm{q}\to \infty$ and $\norm{{p}}\to \infty$.
	\end{proof}
	
	Note that from \eqref{cond1} and \eqref{vcond2}, we get that
	\begin{equation*}
		\dot{S}\leq -\frac{2\beta_{\max}\mu}{1+\epsilon\norm{A}\beta_{\max}}S.
	\end{equation*}
	Furthermore, from the comparison lemma (see \cite{khalil2002nonlinear}), we get
	\begin{equation*}
		S\leq S_0\exp^{-\frac{2\beta_{\max}\mu}{1+\epsilon\norm{A}\beta_{\max}}t},
	\end{equation*}
	where $t\geq0$ is the time variable and $S_0$ is the Lyapunov function \eqref{candidate} evaluated at $t=0$. Then, we get that
	\begin{equation}\label{trajrate}
		\norm{x}\leq \sqrt{\frac{k_2}{k_1}}\norm{x_0}\exp^{-\frac{\beta_{\max}\mu}{1+\epsilon\norm{A}\beta_{\max}}t},
	\end{equation} 
	where $x_0\in \rea^{2n}$ corresponds to the initial conditions vector. Therefore, we establish the following result. 
	\begin{corollary}\label{cor1}
		The trajectories of \eqref{phsys} converge to the desired equilibrium $x_\star$ with a rate of convergence given by 
		\begin{equation}\label{rate}
			\frac{\beta_{\max}\mu}{1+\epsilon\norm{A}\beta_{\max}}.
		\end{equation}
	\end{corollary}
	
	\begin{remark}\label{rem1}
		To ease the presentation of the results of Section \ref{apps}, we have selected $\Phi(q):=A^{\top}(q)$ in the proof of Theorem \ref{prop1}. We remark that this selection is not unique. Indeed, another interesting option that verifies \eqref{condprop1} is $\Phi(q):=A^{-1}(q)$.
	\end{remark}

	\section{Exponential stabilization via PID-PBC}\label{apps}
	The stabilization of mechanical systems via PBC techniques has been extensively studied. In particular, PID-PBC \cite{zhang2017pid,borja2020,romero2018global} represents a constructive methodology to stabilize mechanical systems without solving partial differential equations. Moreover, in this approach, the control parameters may admit a physical interpretation.
	
	In this section, we apply the analysis proposed in Section II to establish conditions that guarantee the exponential stability of mechanical systems stabilized via PID-PBC of \cite{zhang2017pid,borja2020}. Towards this end, we consider mechanical systems that admit a pH representation of the form
	\begin{equation}\label{sysmec}
		\arraycolsep=1pt \def\arraystretch{1.4}
		\begin{array}{rcl}
			&\begin{bmatrix}
				\dot{\mathbb{q}}\\\dot{\mathbb{p}}
			\end{bmatrix} = \begin{bmatrix}0_{n\times n}&I_n\\-I_n &-\mathbb{D}(\mathbb{q,p})\end{bmatrix}
			\begin{bmatrix}\nabla_{\mathbb{q}}\mathbb{H}(\mathbb{q},\mathbb{p}) \\ \nabla_{\mathbb{p}}\mathbb{H}(\mathbb{q},\mathbb{p})\end{bmatrix}+ \begin{bmatrix}0_{n\times n} \\ G	\end{bmatrix}u \\
			&\mathbb{H}(\mathbb{q},\mathbb{p})=\displaystyle\frac{1}{2}\mathbb{p}^\top M^{-1}(\mathbb{q}) \mathbb{p}+\mathbb{U}(\mathbb{q}), ~
			y=G^\top M^{-1}{\mathbb{(q)p}}
		\end{array}
	\end{equation} 
	where $\mathbb{q},\mathbb{p} \in \rea^{n}$ are the generalized positions and momenta vectors, respectively, $\mathbb{H}:\rea^n\times\rea^n\to \rea_+$ is the Hamiltonian of the system, $\mathbb{U}:\rea^n\to\rea_+$ is the potential energy of the system, ${M:\rea^n\to\rea^{n\times n}}$ is the so-called mass inertia matrix, which is positive definite, ${\mathbb{D}:\rea^n\times\rea^n\to \rea^{n\times n}}$ is positive semi-definite and represents the natural damping of the system, $u,y \in \rea^{m}$ are the control and passive output vectors, respectively, $m\leq n$, and $G\in \rea^{n\times m}$ is the input vector with $rank(G)=m$, which is defined as 
	\begin{equation}
		G:=\begin{bmatrix}
			0_{\ell\times m}\\I_m
		\end{bmatrix}, \quad \ell:=n-m.
	\end{equation}
	The set of assignable equilibria for \eqref{sysmec} is defined by
	\begin{equation*}
		\mathcal{E}:=\{\mathbb{q},\mathbb{p} \in \rea^n \ | \ \mathbb{p}=0_n, \ G^\perp \nabla\mathbb{U(q)}=0_{\ell}\},
	\end{equation*}
	where $G^{\perp} := \left[I_{\ell} \ 0_{\ell\times m}\right]$.
		
	Proposition \ref{prop3} establishes that the PID-PBC proposed in \cite{borja2020} preserves the mechanical structure for the closed-loop system. This result is essential in the stability analysis presented in this section.
	\begin{proposition}\label{prop3} 
		Consider a mechanical system represented by \eqref{sysmec}, and the desired configuration $q_{\star}\in\rea^{n}$, such that $(q_{\star},0_{n})\in\mathcal{E}$. Define the PID-PBC controller
		\begin{equation}\label{pidpbc}
			u=-K_P y-K_I(G^\top \mathbb{q}+\kappa)-K_D\dot{y}
		\end{equation}
		where $K_I\in \rea^{m\times m}$ is positive definite, $K_P, K_D\in \rea^{m\times m}$ are positive semi-definite matrices, and $\kappa \in \rea^m$ is defined as
		\begin{equation}\label{kappa}
			\kappa:= -G^\top q_{\star} - K_{I}^{-1}G^\top\nabla\mathbb{U}(q_{\star}).
		\end{equation} 
		Then, the closed-loop system has a stable equilibrium point at $(q_{\star},0_{n})$ if there exists $K_{I}>0$ such that
		\begin{equation}\label{pidcond}
			\nabla^{2}\mathbb{U}(q_{\star})+GK_{I}G^\top>0.
		\end{equation} 
		Moreover, the closed-loop system takes the form 
		\begin{equation}\label{tg1}
			\begin{bmatrix}
				\dot{\mathbb{q}}\\\dot{\mathbb{p}}
			\end{bmatrix}=\mathbb{F}_{d}(\mathbb{q},\mathbb{p})\nabla \mathbb{H}_{d}(\mathbb{q},\mathbb{p})
		\end{equation} 
		with
		\begin{equation}\label{tg2}
			\arraycolsep=2pt \def\arraystretch{1.4}
			\begin{array}{rcl}
				\mathbb{F}_{d}(\mathbb{q},\mathbb{p})&:=& \begin{bmatrix}
					0_{n\times n}&M^{-1}(\mathbb{q})M_d(\mathbb{q})\\ -M_d(\mathbb{q})M^{-1}(\mathbb{q})& \mathbb{J(q,p)}-\mathbb{D}_{d}(\mathbb{q,p})
				\end{bmatrix}\\[0.5cm]	
				\mathbb{H}_d\mathbb{(q,p)}&:=& \displaystyle\frac{1}{2}\mathbb{p}^\top M_{d}^{-1}(\mathbb{q})\mathbb{p}+\mathbb{U}_d\mathbb{(q)}\\
				\mathbb{U}_{d}\mathbb{(q)}&=&\displaystyle\frac{1}{2}(G^\top \mathbb{q}+\kappa)^\top K_I (G^\top \mathbb{q}+\kappa)+\mathbb{U(q)} \\[0.2cm]
				M_d(\mathbb{q})&=&M(\mathbb{q})\left[M(\mathbb{q})+GK_DG^\top\right]^{-1}M(\mathbb{q})\\
				\mathbb{J(q,p)}&=&\mathbb{E}^{-1}\mathbb{(q)}\left[\mathbb{B^\top(q,p)}-\mathbb{B(q,p)}\right]\mathbb{E^{-\top}(q)}\\
				\mathbb{D}_{d}\mathbb{(q,p)}&=&\mathbb{E}^{-1}\mathbb{(q)}(\mathbb{D}+GK_PG^\top)\mathbb{E^{-\top}(q)},\\
				\mathbb{B(q,p)}&:=&GK_D\left( \nabla_{\mathbb{q}} y \right)^{\top},~
				\mathbb{E(q)}:=M(\mathbb{q})M_d^{-1}(\mathbb{q}).		
			\end{array}
		\end{equation} 
	\end{proposition}
	$\square$\\	
	%%%%%%%%%%%%%%%%%%%%%
	\begin{proof}
		Substituting \eqref{pidpbc} into \eqref{sysmec}, we obtain
		\begin{equation}\label{eq1pro3}
			% \arraycolsep=1pt \def\arraystretch{1.4}
			%  \begin{array}{rcl}
			\begin{bmatrix}
				\dot{\mathbb{q}}\\\dot{\mathbb{p}}
			\end{bmatrix}=\bar{\mathbb{F}}\begin{bmatrix}
				\nabla_{\mathbb{q}}\bar{\mathbb{H}} \\ \nabla_{\mathbb{p}}\bar{\mathbb{H}} 		                                
			\end{bmatrix}-\begin{bmatrix}
				0_{n\times m} \\ G
			\end{bmatrix}K_{D}\dot{y}
		\end{equation} 
		where
		\begin{equation*}
			\begin{array}{lll}
				\bar{\mathbb{F}}(\mathbb{q,p})&:=&\begin{bmatrix}
					0_{n\times n} & I_{n} \\ -I_{n} & -\left( \mathbb{D} + GK_{P}G^{\top} \right)
				\end{bmatrix} \\[0.3cm]
				\bar{\mathbb{H}}(\mathbb{q,p})&:=&  \mathbb{U}_d + \displaystyle\frac{1}{2}\mathbb{p}^{\top}M^{-1}(\mathbb{q})\mathbb{p}.                 
			\end{array}
		\end{equation*} 
		Note that
		\begin{equation}\label{eq2pro3}
			\begin{bmatrix}
				\nabla_{\mathbb{q}}\mathbb{H}_{d} \\ \nabla_{\mathbb{p}}\mathbb{H}_{d}
			\end{bmatrix}
			= \begin{bmatrix}
				I_{n} & \left(\nabla_{\mathbb{q}}y\right) K_{D}G^{\top}\\ 0_{n\times n} & I_{n}+\left(\nabla_{\mathbb{p}}y\right) K_{D}G^{\top} 
			\end{bmatrix} \begin{bmatrix}
				\nabla_{\mathbb{q}}\bar{\mathbb{H}} \\ \nabla_{\mathbb{p}}\bar{\mathbb{H}}
			\end{bmatrix}
		\end{equation} 
		with $\nabla_{\mathbb{p}}y = M^{-1}(\mathbb{q})G$.
		Moreover, some manipulations show that $\mathbb{E}(\mathbb{q}) = I_{n}+GK_{D}G^{\top}M^{-1}(\mathbb{q}),$
		% 	\begin{equation*}
		% 		\mathbb{E} = I_{n}+GK_{D}G^{\top}M^{-1},
		% 	\end{equation*} 
		which has \textit{full rank}.
		Hence, \eqref{eq2pro3} can be rewritten as
		\begin{equation}\label{eq3pro3}
			\nabla\mathbb{H}_{d} = \Gamma^{\top}\mathbb{(q,p)} \nabla\bar{\mathbb{H}}, ~		\Gamma\mathbb{(q,p)}: = \begin{bmatrix}
				I_{n} & 0_{n\times n} \\  \mathbb{B(q,p)} & \mathbb{E(q)}                
			\end{bmatrix}.
		\end{equation}
		Note that, since $\mathbb{E}(\mathbb{q})$ has full rank, $\Gamma(\mathbb{q,p})$ is invertible.
		On the other hand,
		\begin{equation*}
			\dot{y} = \left( \nabla_{\mathbb{q}} y \right)^{\top}\dot{\mathbb{q}} + \left( \nabla_{\mathbb{p}} y \right)^{\top}\dot{\mathbb{p}}.
		\end{equation*}
		Thus, \eqref{eq1pro3} is equivalent to
		\begin{equation}
			\Gamma\begin{bmatrix}
				\dot{\mathbb{q}}\\\dot{\mathbb{p}}
			\end{bmatrix}=\bar{\mathbb{F}}\nabla\bar{\mathbb{H}}.
		\end{equation} 
		Therefore, from \eqref{eq3pro3}, it follows that
		\begin{equation}
			\begin{array}{lll}
				\begin{bmatrix}
					\dot{\mathbb{q}}\\\dot{\mathbb{p}}
				\end{bmatrix}=\Gamma^{-1}\bar{\mathbb{F}}\Gamma^{-\top}\nabla{\mathbb{H}}_{d} ={\mathbb{F}}_{d}\nabla{\mathbb{H}}_{d}.
			\end{array}
		\end{equation} 
		Note that, ${\mathbb{F}}_{d_{sym}}(\mathbb{q,p})\leq 0$. Thus, $\dot{\mathbb{H}}_{d}\leq 0$.
		Additionally,
		\begin{eqnarray}
			\nabla\mathbb{H}_{d}(q_{\star},0_{n})& = &\begin{bmatrix}
				\mathbb{\nabla_{q}H}_{d}(q_{\star},0_{n}) \\ 0_{n}                         
			\end{bmatrix}=0_{2n}	\label{eq5pro2} 
		\end{eqnarray} 
		where we have used \eqref{kappa}. Furthermore,
		\begin{equation*}
			\nabla^{2}_\mathbb{{q}}\mathbb{H}_{d}(q_{\star},0_{n}) = \begin{bmatrix}
				\nabla^{2}\mathbb{U}(q_{\star})+K_{I} & 0_{n\times n} \\  0_{n\times n} & M_{d}^{-1}(q_{\star})                                      
			\end{bmatrix}.
		\end{equation*} 
		Hence, \eqref{pidcond} guarantees that
		\begin{equation}\label{eq6pro2}
			\nabla^{2}_\mathbb{{q}}\mathbb{H}_{d}(q_{\star},0_{n})>0.
		\end{equation} 
		The expressions \eqref{eq5pro2} and \eqref{eq6pro2} imply that $\mathbb{H}_{d}(\mathbb{q,p})$ has a locally isolated minimum at $(q_{\star},0_{n})$, which in combination with \eqref{eq3pro3} prove the stability of the equilibrium for the closed-loop system.
	\end{proof}
	\begin{remark}
		For \textit{fully actuated} mechanical systems, the control law \eqref{pidpbc} can be modified as follows
		\begin{equation*}
			u=\nabla_{\mathbb{q}}\mathbb{U(q)}-K_P y-K_I(G^\top \mathbb{q}-q_{\star})-K_D\dot{y},
		\end{equation*} 
		where the first term compensates the gravity effects.
	\end{remark}
	\begin{remark}\label{pidclass}
		The PID-PBC scheme is applied to the passive output signal--which for mechanical systems correspond to the actuated velocities. On the other hand, the classical PID controller is applied to an error signal--which for mechanical systems is customarily given by the error between the actual and the desired position of the system. For some cases, a PI-PBC scheme coincides with the classical PD controller.
	\end{remark}
	\subsection{Proving the ES properties of the closed-loop system}\label{sec:ES}
	As it is shown in the previous section, the stabilization of a mechanical system via PID-PBC yields a new mechanical system described by \eqref{tg1}-\eqref{tg2}. In this section, we prove that, under some mild conditions, $(q_\star,0_n)$ is an \textit{exponentially stable} equilibrium point for the resulting closed-loop mechanical system. To this end, we introduce the change of coordinates described in \cite{venkatraman2010speed} 
	\begin{equation}\label{change}
		p:= T_d^\top(\mathbb{q})\mathbb{p},\quad q:=\mathbb{q}-q_\star,
	\end{equation}
	where $T_d:\rea^{n}\to \rea^{n\times n}$ is the upper Cholesky factor of $M_d^{-1}(\mathbb{q})$, i.e., $T_d(\mathbb{q})$ is a \textit{full rank} \textit{upper triangular} matrix with strictly positive diagonal entries such that $	M_d^{-1}(\mathbb{q})=T_d(\mathbb{q})T_d^\top (\mathbb{q})$. Hence, by using the change of coordinates \eqref{change} and the results of Theorem \ref{prop1}, the following proposition establishes conditions that guarantee that $(q_\star,0_n)$ is an exponentially stable equilibrium point for \eqref{tg1}-\eqref{tg2}.		
	
	%%%%%%%%%
	\begin{proposition}\label{prop4}
		Consider the system \eqref{tg1}-\eqref{tg2}. Then, its equilbrium point $(q_\star,0_n)$ is ES if:
		\begin{itemize}[wide=0pt]
			\item[\textbf{C1}] $\mathbb{U}_{d}(\mathbb{q})$ is \textit{strongly convex}.
			\item[\textbf{C2}] $\lVert M^{-1}(\mathbb{q})M_d(\mathbb{q}) \rVert<\infty$.
			\item[\textbf{C3}] $\mathbb{D}_{d}(\mathbb{q,p})>0$.
		\end{itemize}
		Moreover, the equilbrium is GES if $\mathbb{U}_{d}(\mathbb{q})$ is \textit{radially unbounded}.
	\end{proposition}
	%%%%%%%%%%%%%%%%%%%%
	\begin{proof}
		Note that the change of coordinates \eqref{change} transforms \eqref{tg1}-\eqref{tg2} into \eqref{phsys}, with 
		\begin{equation}\label{def2}
			\begin{split}
				U(q):=&\mathbb{U}_{d}(q+q_\star),\\
				A(q):=&\bar{M}^{-1}(q)\bar{T}_d(q)^{-\top},\\
				D(q,p):=&\bar{T}_d^\top(q) \bar{\mathbb{D}}_d{(q,p)} \bar{T}_d(q),\\
				J(q,p):=&J_3(q,p)+\bar{T}_d^\top(q) \bar{\mathbb{J}}(q,p)\bar{T}_d(q),\\
				J_3(q,p):=&\sum_{i=1}^n\Bigg[\Bigg({p}^\top \bar{T}_d^{-1}(q)\frac{\partial \bar{T}_d(q)}{\partial q_i}\Bigg)^\top (A^\top(q) e_i)^\top,\\
				&-(A^\top(q) e_i)\Bigg(\;{p}^\top \bar{T}_d^{-1}(q)\frac{\partial \bar{T}_d(q)}{\partial q}\Bigg)\Bigg],\\
			\end{split}
		\end{equation}
		where 
		$$
		\begin{array}{lll}
			&\bar{M}(q):=M(q+q_\star),\bar{\mathbb{J}}(q,p):=\mathbb{J}(q+q_\star,\bar{T}_d^{-\top} (q)p)\\
			&\bar{T}_d(q):=T_d(q+q_\star),\bar{\mathbb{D}}_d(q,p):=\mathbb{D}_d(q+q_\star,\bar{T}_d^{-\top}(q)p).
		\end{array}
		$$
		On the other hand, \textbf{C1} implies that Assumption \ref{ass1} is satisfied. Moreover, from \textbf{C2}, we have the following chain of implications
		\begin{equation*}
			\begin{array}{lll}
				&\lVert M^{-1}(\mathbb{q})M_d(\mathbb{q}) \rVert<\infty \implies \lVert M_d(\mathbb{q}) \rVert<\infty \\
				&\implies \lVert T_d(\mathbb{q}) \rVert<\infty \implies  \lVert  \bar{T}_d(q) \rVert<\infty \\
				&\implies \lVert A(q) \rVert<\infty.
			\end{array}
		\end{equation*} 
		Thus, Assumption \ref{ass2} is satisfied. The rest of the proof follows from Theorem \ref{prop1}, noting that \textbf{C3} implies that $D(q,p)>0$.
	\end{proof}
	
	\begin{remark}
		\textbf{C2}--thus, Assumption \ref{ass2}--is not restrictive from a physical point of view. Some simple computations show that \textbf{C2} reduces to $\lVert M(\mathbb{q})\rVert<\infty$. The class of robot manipulators that verify this condition is provided in \cite{ghorbel1993positive}.
	\end{remark}
	
	\begin{remark}
		The term $C(q,p)$ can be regarded as a \textit{virtual energy} term. Particularly, for ${\Phi(q)=A^\top(q)}$, such a term reduces to $C(q,p) = \epsilon \dot{U}_d(q)$.
		% 	\\\centerline{$C(q,p)= \epsilon \dot{q}^\top \nabla_q U(q) = \epsilon \dot{U}_d(q)$}.\\	
		Hence, considering $\epsilon$ as a constant with unit of seconds $[s]$, we have that $C(q,p)$ is expressed in units of Joules $[J]$ (energy).
	\end{remark}
	
	\subsection{Discussion}
	In Theorem \ref{prop1} it is required that $D(q,p)>0$. Note that, this condition is translated to \textbf{C3} in Proposition \ref{prop4}. For some particular cases, additional constraints on the natural damping of the system $\mathbb{D(q,p)}$ are needed to verify the mentioned requirement.  From a physical perspective, \textbf{C3} is not restrictive as  dissipation--in this case the natural damping--is inherent to the nature of mechanical systems. However, this physical phenomenon is usually neglected to simplify the mathematical modeling of the system under study. Then, to determine the range of applicability of the ES analysis exposed in Section \ref{sec:ES}, we analyze two particular cases of interest:
	\begin{enumerate}[label=(\roman*),wide=0pt]
		\item \textbf{Fully-actuated mechanical systems ($m=n$)}: if $\mathbb{D}\mathbb{(q,p)}=0_{n\times n}$, then the control parameters are chosen such that $\mathbb{D}_d\mathbb{(q,p)}>0$, i.e., $K_P>0$ . Hence, the application of results given in Proposition \ref{prop4} is straightforward.
		\item \textbf{Underactuated mechanical systems ($m<n$)}: when the open-loop system is underactuated, it is necessary to impose some conditions on $\mathbb{D}\mathbb{(q,p)}$ to ensure that $\mathbb{D}_d\mathbb{(q,p)}>0$. Such a condition is $G^\perp\mathbb{D(q,p)} (G^\perp)^\top>0.$
		Then, for $K_{P}>0$, we have
			$$\mathbb{D}\mathbb{(q,p)} + GK_{P}G^{\top}>0 \implies \mathbb{D}_d\mathbb{(q,p)}>0.$$
			\vspace{-4mm}
	\end{enumerate}

	\section{A tuning guideline}\label{sec:tune}
	As shown in the ES analysis provided in Section \ref{main}, the trajectories of systems represented as in \eqref{phsys} have a rate of convergence given by \eqref{rate}. Hence, by proving ES properties for \eqref{tg1}-\eqref{tg2} via such an analysis, the rate of convergence of mechanical systems stabilized by PID-PBC is also given by \eqref{rate}. 
	
	Since the terms $\beta_{max}$, $\norm{A(q)}$, and $\mu(q,p)$ are associated directly with the potential energy, the kinetic energy, and the damping of the system, respectively,  \eqref{rate} provides insight into how the mentioned physical quantities affect the rate of convergence. This intuition can be exploited to design the control parameters of PID-PBC.  Thus, by considering that each control parameter from these methodologies is associated with at least one physical property of the closed-loop system, we investigate the effect of such parameters on the rate of convergence via the expression \eqref{rate}. Accordingly, we can select the control parameters related to PID-PBC such that closed-loop system \eqref{tg1}-\eqref{tg2} has a prescribed performance in terms of its rate of convergence. 
	
	Note that the rate of convergence of the closed-loop system is given in terms of four elements, namely, $\beta_{\max}$, $\norm{A(q)}$, $\mu(q,p)$, and $\epsilon$. Therefore, a guideline to establish a relationship between these elements and the control parameters of PID-PBC is given as follows:
	\begin{itemize}[wide=0pt]
		\item $\beta_{max}$: the rate of convergence is proportional to $\beta_{max}$ since $\beta_{\max}:=\max\{1,\lambda_{\max}(G K_I G^\top)\}.$
		\item $\norm{A}$: note that from \eqref{def2}, we get
		\begin{equation}\label{norma}
			\norm{A}=\norm{\bar{M}^{-1}\bar{T}_d^{-\top}}\leq \norm{\bar{M}^{-1}}\norm{\bar{T}_d^{-\top}}.
		\end{equation}
		Therefore, from \eqref{rate} and \eqref{norma}, the rate of convergence can be increased by increasing $\norm{\bar{T}_d^{-\top}}$. This term is related to $K_D$, see \eqref{tg2}.
		\item  $\epsilon$ and $\mu(q,p)$:  the expression \eqref{rate} provides intuition of the effect of the control parameters even without performing the cumbersome computations of $\epsilon$ and $\mu(q,p)$. For example, the Schur complement of \eqref{upsilonsym2} is given by $Z-Y^\top X^{-1} Y$ with
		\begin{equation}
			\begin{split}
				X(q)&:=\epsilon \bar{M}^{-1}(q)\bar{M}_d(q)\bar{M}^{-1}(q)\\
				Y(q,p)&:=\frac{\epsilon}{2}[A(q)(D(q,p)-J(q,p))-\dot{A}(q)]\\
				Z(q,p)&:=D-\epsilon(A^\top(q) \nabla_q^2U(q)A(q)).
			\end{split}
		\end{equation}
		Note that $Y(q,p)$ increases as $D(q,p)$ increases. Then, it follows that $\epsilon$ must be reduced to guarantee that $Z(q,p)>0$ and to ensure that the Schur complement is positive definite as well.  
	\end{itemize}

	\section{Experimental Results}
	\begin{table}[t]
		\caption{Tuning Gains}\label{gains}
		\centering
		\vspace{-3mm}
	\begin{tabular}{llll}
	\hline
	& \multicolumn{1}{c}{$K_P$} & \multicolumn{1}{c}{$K_I$} \\ \hline
	S1 & diag(5,15,20)            & diag(200,250,350)         \\
	S2 & diag(5,15,20)            & diag(200,250,200)     \\
	S3 & diag(5,1,20)            & diag(200,250,350)         \\ \hline
\end{tabular}
	\end{table}
	In this section, we illustrate the applicability of \eqref{rate} as a tuning guideline by showing the effect of modifying $\beta_{\max}$ and $\mu(q,p)$ on the rate of convergence of a mechanical system.  To this end, we implement a PID-PBC to stabilize the Philips Experimental Robotic Arm (PERA), depicted in Fig.~\ref{pera}. The PERA is a seven degrees-of-freedom (DoF) experimental robotic arm created by Philips Applied Technologies \cite{rijs2010philips} to mimic the motion of a human arm. To ease the presentation of our results, we reduce the model to three DoF, namely, 
		\begin{itemize}[wide=0pt]
		\item The yaw shoulder joint $Y_S$ with angle $\mathbb{q}_1$.
		\item The pitch elbow joint $P_E$ with angle $\mathbb{q}_2$.
		\item The yaw elbow joint $Y_E$ with angle $\mathbb{q}_3$.
	\end{itemize}
	\begin{figure}[t]
		\centering
		\includegraphics[width=0.2\textwidth]{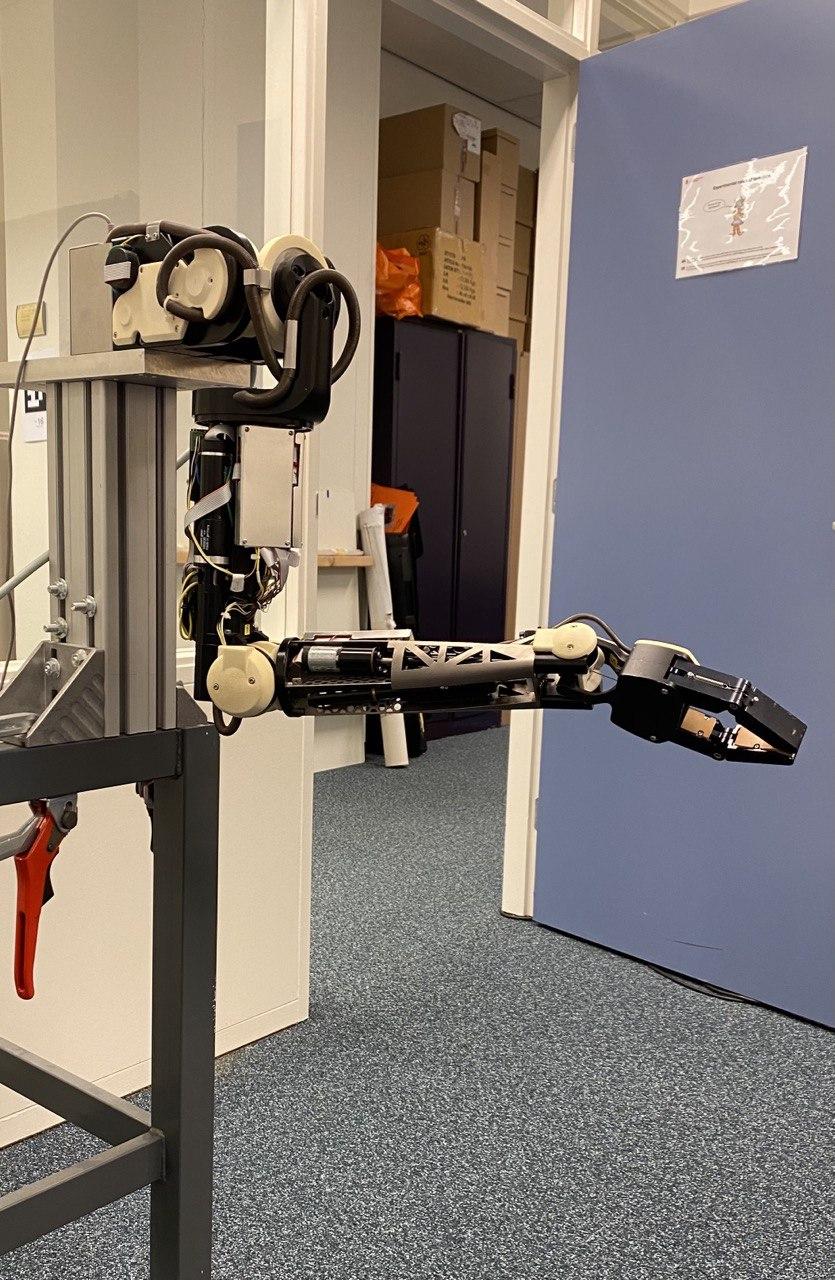}
		\vspace{-2mm}
		\caption{Experimental Setup: Philips Experimental Robotic Arm (PERA)}\label{pera}
			\vspace{-2mm}
	\end{figure}
	The PERA system can be modeled by \eqref{sysmec}, with $n=m=3$, $G=I_3$,$\mathbb{U(q)}=m_2d_{c2}g(1-\cos(\mathbb{q}_2))$, and 
	\begin{equation*}
		\begin{array}{ll}
				M(\mathbb{q})=\begin{bmatrix}
				m_{11}&0&m_{13}\\
				0&m_{22}&0\\
				m_{13}^\top&0&m_{33}
			\end{bmatrix}
		\end{array}
	\end{equation*}
	with
		\begin{equation*}
			\begin{array}{ll}
				m_{11}:=\displaystyle\sum_{i=1}^{3}I_{i}+m_2d_{c2}^2\sin^2(\mathbb{q}_2), &
			m_{13}:=I_3\cos(\mathbb{q}_2),\\
					m_{22}:=\displaystyle\sum_{j=2}^{3}I_{j}+m_2d_{c2}^2, &
				m_{33}:=I_3,
			\end{array}
	\end{equation*}
	where $I_1$, $I_2$, and $I_3$ correspond to the moments of inertia of the joints $Y_S$, $P_E$, and $Y_E$, respectively\footnote{These values are not relevant for the analysis provided in this manuscript.}, $m_2=1\ kg$ is the mass of the link composed of the elbow and wrist, $d_{c2}=0.16\ m$ is the distance to the center of mass of $m_2$, and $g=9.81 m/s^2$ is the gravity.
	We stabilize the PERA at the desired configuration ${q_\star=col(-1.8, 1.57, 0.78)}$ with three set of tuning gains, namely, $S_1$, $S_2$, and $S_3$. In the three cases, we select $K_D=0_{3\times 3}$, which corresponds to a PI-PBC scheme. The rest of the gains are shown in Table \ref{gains}. 

	A video of the experimental results can be found in: \url{https://youtu.be/-ty0D8VKQMs}.
	\begin{figure}[t]
	\centering
	\includegraphics[width=\columnwidth]{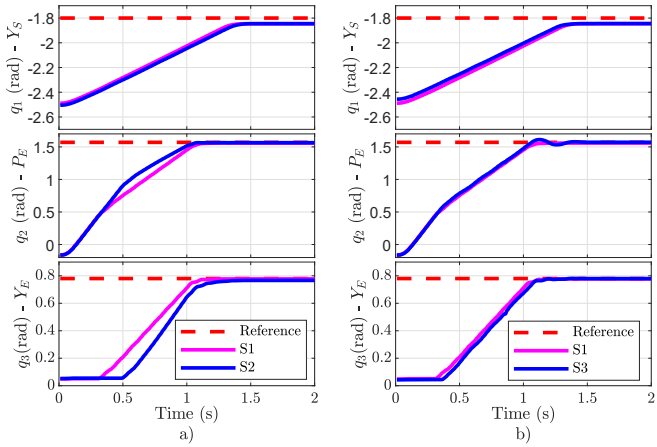}
	\caption{Comparison of trajectories of angular position: \\a) S1 vs. S2. b) S1 vs. S3.}	\label{res}
	\end{figure}
	Fig. \ref{res}.a shows the effect of increasing $\beta_{\max}$, in this case $K_I$. Note that in \eqref{rate}, by fixing the parameters $\mu(q,p)$, $\epsilon$, and $\norm{A(q)}$, the rate of convergence increases as $\beta_{max}$ increases. Therefore, it is expected that the rate convergence of S1 is faster than the rate of convergence of S2, which is verified, particularly in $Y_E$, by the results depicted in the mentioned figure.
	 Fig. \ref{res}.b illustrates the effect of modifying $K_P$, which is directly proportional to $\mu(q,p)$. Therefore, it follows that the rate of convergence of S1 is faster than the rate of convergence of S3, which is verified, particularly for $P_E$ and $Y_E$, by the results shown in the mentioned figure.
	There is a small steady-state error in the joint positions, particularly for $Y_S$. This error may be due to non-modeled physical phenomena such as dry friction or asymmetry of the motors (for further details, see Remark \ref{pidclass}).
	%\begin{figure}[t]
	%	\includegraphics[width=\columnwidth]{fig/results.eps}
	%	\caption{Trajectories for angular position}	\label{damp}
	%\end{figure}
	
	\section{Concluding Remarks}
	In this paper, we have presented an analysis to demonstrate ES properties for a class of pH systems. Furthermore, we have proven that such an analysis is suitable, under some mild conditions, to show ES properties for nonlinear mechanical systems stabilized via PID-PBC. Moreover, with the proposed Lyapunov candidate function, we have established a relationship between the physical quantities--i.e., damping and energy--and the rate of convergence of the closed-loop system. Since the PBC techniques control parameters are associated with the energy shaping process and the damping injection process, we have endowed with physical intuition the process of control parameters selection to assign a performance to the system \textit{in terms of its rate of convergence}.

	\bibliographystyle{ieeetr}
	\bibliography{ref} 
\end{document}